\documentclass[11pt]{article}
\usepackage[utf8]{inputenc}
\usepackage[american]{babel}
\usepackage[T1]{fontenc}
\usepackage{amssymb}
\usepackage{amsmath}
\usepackage{listings}
\usepackage{graphicx}
\usepackage{color}
\usepackage{xspace}
\usepackage{todonotes}
\usepackage{enumerate}
\usepackage{xargs}
\usepackage{tabularx}
\usepackage{multirow}
\usepackage{listings}
\usepackage{stmaryrd}
\usepackage{mathtools}  % :=, \mathclap and several bugfixes for amsmath
\usepackage[backend=bibtex,doi=false,isbn=false,url=false,firstinits=true,maxcitenames=2,maxbibnames=99]{biblatex}

\usepackage[plainpages=false]{hyperref}
\hypersetup{
	unicode=false,
	colorlinks=true,
	linkcolor=blue,
	citecolor=blue,
	filecolor=blue,
	urlcolor=blue,
	pdfauthor={Sebastian Abshoff, Andreas Cord-Landwehr, Daniel Jung, Friedhelm Meyer auf der Heide},
	pdfsubject={Towards Gathering Robots with Limited View in Linear Time: The Closed Chain Case},
	pdftitle={Towards Gathering Robots with Limited View in Linear Time: The Closed Chain Case},
	pdfkeywords={autonomous robots, gathering problem, runtime bound, robot formation problems, distributed algorithms},
}

\AtEveryBibitem{
    \clearlist{location}
    \clearfield{pages}
    \clearfield{url}
    \clearname{editor}
}

\newtheorem{theorem}{Theorem}[section]
\newtheorem{lemma}[theorem]{Lemma}

\newtheorem{definition}[theorem]{Definition}
\newtheorem{proposition}[theorem]{Proposition}

\newenvironment{proof}{\noindent {\bf Proof.}\ }{\qed\par\vskip 4mm\par}
\newenvironment{proof*}{\noindent {\bf Proof.}\ }{\par\vskip 4mm\par}

\newcommand{\qed}{\hfill $\square$}

\newcommandx*{\LDAUOmicron}[2][1=@pkling_false]{\mathcal{O}\ifthenelse{\equal{#1}{small}}{\bigl(#2\bigr)}{\left(#2\right)}}
\newcommandx*{\LDAUomicron}[2][1=@pkling_false]{\mathrm{o}\ifthenelse{\equal{#1}{small}}{\bigl(#2\bigr)}{\left(#2\right)}}
\newcommandx*{\LDAUOmega}[2][1=@pkling_false]{\Omega\ifthenelse{\equal{#1}{small}}{\bigl(#2\bigr)}{\left(#2\right)}}
\newcommandx*{\LDAUomega}[2][1=@pkling_false]{\omega\ifthenelse{\equal{#1}{small}}{\bigl(#2\bigr)}{\left(#2\right)}}
\newcommandx*{\LDAUTheta}[2][1=@pkling_false]{\Theta\ifthenelse{\equal{#1}{small}}{\bigl(#2\bigr)}{\left(#2\right)}}

\newcommandx*{\set}[2][2=@pkling_false]{\left\{#1\ifthenelse{\equal{#2}{@pkling_false}}{}{\;\middle|\;#2}\right\}}

\newcommand{\OBdA}{W.l.o.g.\xspace}

\newcommand{\calO}{\mathcal{O}}

\newcommand{\dist}[1]{\mathrm{dist}(#1)}
\newcommand{\timestamp}[1]{\mathrm{timestamp}(#1)}

\newcommand{\subchain}[1]{\mathrm{subchain}(#1)}

\newcommand{\EM}{\mathrm{EM}}
\newcommand{\VM}{\mathrm{VM}}

\newcounter{algorithmNumber}
\newcommand{\algorithm}[1]{%
    \medskip\par%
    \noindent%
    \refstepcounter{algorithmNumber}%
    \textbf{Algorithm~\thealgorithmNumber~(\texttt{#1()}).\ }%
}

\newcounter{actionNumber}
\DeclareRobustCommand{\action}[1]{%
   \refstepcounter{actionNumber}%
   $A_{\theactionNumber\label{#1}}$}

\newcounter{symactionNumber}

\makeatletter
\xspaceaddexceptions{\check@icr}
\makeatother

\title{%
Towards Gathering Robots with Limited View in Linear Time: The Closed Chain Case%
\thanks{This work was partially supported by the International Graduate School ``Dynamic Intelligent Systems''.}}

\author{
    Sebastian Abshoff
    \and
    Andreas Cord-Landwehr
    \and
    Daniel Jung
    \and
    Friedhelm Meyer auf der Heide
}

\date{
    Heinz Nixdorf Institute \& Computer Science Department\\[0.2em]
    University of Paderborn (Germany)\\[0.2em]
    Fürstenallee 11, 33102 Paderborn\\[0.2em]
\quad\\[0.2em]
    \texttt{\{abshoff,cola,daniel.jung,fmadh\}@upb.de}
}

\bibliography{references}

\begin{document}

\maketitle

\thispagestyle{empty}

\begin{abstract}
In the gathering problem, $n$ autonomous robots have to meet on a single point.
We consider the gathering of a closed chain of point-shaped, anonymous robots on a grid.
The robots only have local know\-ledge about a constant number of neighboring robots along the chain in both directions.
Actions are performed in the fully synchronous time model $\mathcal{FSYNC}$.
Every robot has a limited memory that may contain one timestamp of the global clock, also visible to its direct neighbors.
In this synchronous time model, there is no limited view gathering algorithm known to perform better than in quadratic runtime.
The configurations that show the quadratic lower bound are closed chains.
In this paper, we present the first sub-quadratic---in fact linear time---gathering algorithm for closed chains on a grid.
\end{abstract}

\section{Introduction}
Over the last years, there was a growing interest in problems related to the creation of formations by autonomous robots.
A benchmark problem is gathering:
Given a configuration of $n$ robots (on the plane or on the grid), they have to gather in one position not specified beforehand.
In this paper, we consider a closed chain of $n$ robots on a two-dimensional grid.
We assume that robots have a bounded viewing range: They can ``see'' the relative positions of a constant number of neighboring robots in both directions on the chain, but have no compass.
We assume the fully synchronous time model ($\mathcal{FSYNC}$).
The robots are anonymous and have a memory limited to save only one timestamp of a global clock which can be seen by their neighbors.
We provide an algorithm for gathering in this model which only needs
$\calO(n)$ rounds. This result is asymptotically optimal for worst case closed chains.

All known gathering algorithms with a limited viewing range, like the local view gathering algorithm by \textcite{MINCH}, only achieve a gathering time of $\calO(n^2)$, and come with matching lower bounds for their strategy.
Interestingly, the configurations for the lower bounds are closed chains.
%\newpage
\paragraph{Related Work}
There is a vast literature on robot problems, researching how specific coordination problems can be solved by a swarm of robots given a certain limited set of abilities.
The robots are usually point-shaped (hence collisions are neglected) and positioned in the Euclidean plane.
They can be equipped with a memory or are \emph{oblivious}, i.e., the robots do not remember anything from the past and perform their actions only on their current views.
If robots are anonymous, they do not carry any IDs and cannot be distinguished by their neighbors.
Another type of constraint is the compass model:
If all robots have the same coordinate system, some tasks are easier to solve than if all robots' coordinate systems are distorted.
\textcite{gathering-compasses} provide a classification of these two and also of dynamically changing compass models, as well as their effects regarding the gathering problem in the Euclidean plane.

The operation of a robot is considered in the common \emph{look-compute-move model} \cite{Cohen:2004a}, which parts one operation in three steps.
In the \emph{look} step the robot gets a snapshot of the current scenario from its own perspective.
During the \emph{compute} step, the robot computes its action, and eventually performs it in the \emph{move} step.
How the steps of several robots are aligned is given by the \emph{time model}, which can range from an asynchronous $\mathcal{ASYNC}$ model (for example, see \cite{Cohen:2004a}), where even the single steps of the robots' steps may be interleaved, to a fully synchronous $\mathcal{FSYNC}$ model (for example, see \cite{localgathering}), where all steps are performed simultaneously.
While $\mathcal{ASYNC}$ is more suited for questions of convergence or termination, the $\mathcal{FSYNC}$ model gives precises runtime bounds when the robots are working completely in parallel.

A collection of recent algorithmic results concerning distributed solving of basic problems like gathering and pattern formation, using robots with very limited capabilities, can be found in \textcite{flocchinioverview}.

An interesting problem under the $\mathcal{FSYNC}$ model for point shaped robots with only local view is the shorting of some given (open) communication chain between
two fixed (unmovable) robots.
In \cite{hopper}, this problem is solved in time $\calO(n)$ for robots on a grid with the \emph{Manhattan Hopper} strategy.
For the Euclidean plane the similar $\emph{Hopper}$ strategy results in a $\sqrt{2}$-approximation, also after time $\calO(n)$.

One of the most natural problems is to gather a swarm of robots on a single point.
Usually, the swarm consists of point-shaped, oblivious, and anonymous robots and the problem is widely studied in the Euclidean plane.
Having point-shaped robots, collisions are understood as merges/fusions of robots and interpreted as gathering progress (cf.\ \textcite{MINCH}).
\textcite{gathering-icalp} provide the first gathering algorithm for the $\mathcal{ASYNC}$ time model with multiplicity detection and global views.
Gathering in the local setting was studied by \textcite{localgathering}.
\textcite{impossibilityofgathering} studied situations when no gathering is possible.
The question of gathering on graphs instead of gathering on the plane was considered by \textcite{practicalrendevouzaktuell}, \textcite{rendezvousingraphen} and \textcite{gatheringOnRing}.
\textcite{gatheringongrids} showed that for gathering on grids multiplicity detection is not needed, and they further provide a characterization of solvable gathering configurations on grids.

Only few results are known about the runtime for gathering algorithms.
In \cite{gatheringthetanquadrat}, a $\calO(n^2)$ runtime bound is shown for an algorithm, working on point shaped robots in the Euclidean plane with constant viewing range in the $\mathcal{FSYNC}$ model.
Here, the robots synchronously compute the center of the smallest enclosing circle of all robots within their viewing range and then move towards its center.
The authors also prove that for their algorithm this bound is tight.

But whether in this model gathering can be solved faster, is still an unanswered question.
Our main goal is to show that in this model gathering is possible in linear time.
In our paper we get closer to this by developing a linear time strategy for gathering of closed chains/loops on a grid.
\paragraph{Model and Notation}
For a set of robots $\mathcal{R} \coloneqq \{1,\ldots,n\}$, every robot $r\in \mathcal{R}$ has exactly two neighbors such that the robots form a closed chain.
We say that two robots $r$ and $r'$ have distance $k$ if the number of robots on the (shorter) chain between $r$ and $r'$ is exactly $k-1$, and the sequence of robots between $r$ and $r'$ (including both) is called a \emph{subchain}.
This distance is denoted by $\dist{r,r'}$ and the \emph{subchain} by $\subchain{r,r'}$.
The sequence of all robots at constant distance to a robot $r$ is called the \emph{viewing range} of $r$.
The robots are positioned on a two-dimensional grid and the $L_1$-distance between a robot and its direct neighbors according to the chain is at most $1$.
Each robot has only a local coordinate system, i.e., a private knowledge about north and west.
For our discussion, we only use north and west in terms of a global coordinate system and map the movements of the robots accordingly.

Time proceeds in fully synchronous rounds ($\mathcal{FSYNC}$).
We group a constant number of rounds to phases and subphases such that a phase consists of a constant number of subphases and a subphase consists of a constant number of rounds.
There exists a global clock, accessible by the robots, which provides the current round number.
In every round, the robots simultaneously consider their local views, compute their next operations, and finally perform their operations.
If after a hop two neighboring robots (according to the chain) are located at the same position, we call this a \emph{merge} and the merged robots are further considered as one robot.
In this case, neighborhoods are merged accordingly and the local coordinate system is arbitrarily selected from one of the merged robots.

Robots can decide to ``run'' along the chain in a specific chain direction, i.e., iteratively swapping the chain order with the next robot in a moving direction (clockwise or counter-clockwise).
They also can perform diagonal hops in order to locally modify the shape of the chain.
We call such a robot a \emph{runner}.
Runners can detect which of the robots in their viewing range are also runners including the relative movement direction on the chain.
Robots have a limited memory to store exactly one timestamp, which is visible to all other robots in the viewing range.
If not mentioned otherwise, the timestamp when the run was started is stored in that memory.

Our strategy will, using the runners, transform the chain such that it can be locally shortened.

For being able to build an algorithm, we disassemble the whole chain into subchains (modules) of specific types and also identify runners with a specific type of a ``moving module'' and then call these modules \emph{runs}.
We will make sure that a run does not overlap with robots of other runs.
The distance $\dist{R,R'}$ of two runs $R$ and $R'$ is the minimal distance $\dist{r,r'}$ of two robots $r\in R$ and $r'\in R'$.
The \emph{subchain} (including both modules, $R$ and $R'$) between $R$ and $R'$ is denoted by $\subchain{R,R'}$.

We say the gathering problem is solved if all robots are included in a square of constant size.
This situation then could be recognized by the robots if we, for example, would extend their capabilities by a constant viewing range on the grid (not only along the chain).
Then gathering could be continued until all robots are located at just a single point or in a square of size 1.
\paragraph{Our Contribution}
In our model, the problem of shortening an open chain between two fixed (unmovable) robots, the first runtime bound of $\calO(n^2\log(n))$ has been shown in \cite{gtm}.
Later, this was improved in \cite{hopper}:
In the Euclidean plane, the \emph{Hopper} strategy delivers a $\sqrt{2}$-approximation of the shortest chain in time $\calO(n)$.
Restricted to a grid, the \emph{Manhattan Hopper} strategy delivers an optimal solution in time $\calO(n)$.

The (general) gathering problem in the Euclidean plane is solved in time $\calO(n^2)$ in \cite{gatheringthetanquadrat} by the \emph{Go-To-The-Center} strategy.
It is still unknown if the gathering problem in this setting can be solved faster.
Our main goal is the improvement of this bound.
In this paper, we start by restricting the robots locations to a grid and just looking at closed chains of robots.
This problem is similar to the open chain problem mentioned above (assume that both endpoints are located at the same position).
But since a closed chain does not have endpoints anymore, several new difficulties arise:
None of the robots are distinguishable and all robots are allowed to move.
We present a strategy which solves the gathering problem in time $\calO(n)$.
Note that closed chains are challenging---some also worst case---inputs for the quadratic \emph{Go-To-The-Center} strategy.
\section{Gathering Algorithm}
The high-level idea of our  gathering algorithm is as follows:
Time is partitioned into phases of constant length and all robots act synchronously. Within a phase, three kinds of actions take place:
\begin{itemize}
\item Merges: If a robot observes to be in a configuration where it can ``merge'' with a neighbor such that the chain remains connected, it merges.
Merged robots behave like one robot from now on. 
\item Starting runs: Robots whose neighborhood has a certain (locally observable) property start running along the chain in a certain direction. They become \emph{runners}.
\item Running: \emph{Runners} move forward along the chain. For this, the runner swaps its position on the chain with the next robot in moving direction.
\end{itemize}
The following properties imply a linear time gathering algorithm:
\begin{itemize}
\item A merge of two robots constitutes a progress of the algorithm, because the chain becomes shorter by one. After $n-1$ merges, gathering is done.
\item In every period of a fixed constant number of phases, in which no merge happens, a ``good pair of runs'' starts. Such good pairs result in a merge after $\calO(n)$ phases. Different good pairs result in different merges. This is what we call \emph{pipelining}.
\end{itemize}
Defining good pairs is crucial. In addition, turning this idea into an algorithm turns out to be challenging. Among, we have to explicitly handle concurrent runs which might influence each other in an undesired way.

Every robot synchronously executes the main Algorithm \ref{algo:gathering} \texttt{gathering()}.
This means that even if during the calls of the subalgorithms (\texttt{merge()}, \texttt{execute-run()}, \ldots) no actions are performed,
the robots wait instead in order to ensure the synchronicity.
Also the subalgorithms are executed fully synchronously.

During their movements \emph{runners} can perform diagonal hops.
This is used for reshaping the chain such that a local merge can be performed afterwards.
In order to find suitable starting points for runners, we cannot rely on the robots themselves, as they are indistinguishable.
Hence, we take local geometric properties of the chain that can be recognized by the affected robots even with their restricted viewing ranges.
A closed chain on a grid can be interpreted as a rectangular polygon with self-intersections.
The edges of such a polygon then correspond to the straight subchains of the configuration.
Taking the geometric shape of such a polygon, the basic idea is to pairwise start runners, moving towards each other, at both ends of every edge.
We will show that if no merge is possible, at least one of the newly started run pairs will reshape the chain this way that some phases later a local merge can be performed.
Such run pairs will be called ``good pairs''.
While moving, a runner performs diagonal hops to a fixed side ($A$ or $B$ in Figure~\ref{fig:doubleturn}.$(a)$) of the polygon's edge.
It turns out that a merge operation can be performed when these two runners meet, in case both runners of such a pair perform these hops to the global same side ($A$ or $B$) of the edge.
This is the case, if during a walk on the polygon's border in arbitrary fixed orientation, the turns at the beginning and end of the edge are performed
in the same direction (cf.\ Figure~\ref{fig:doubleturn}.$(a)$).
\begin{figure}[h]
\centering
    \includegraphics{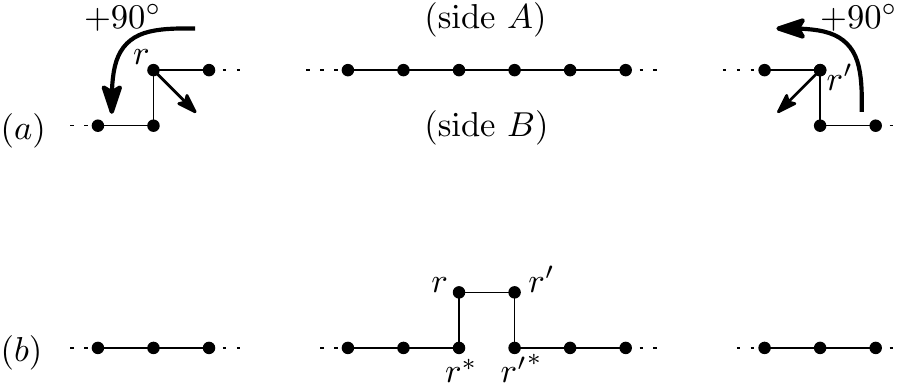}
    \caption{$(a)$: On a walk along its border, two consecutive vertices of a rectangular polygon perform a turn in the same direction ($+90^\circ$). Observe that both runners $r,r'$ initiate their hops to the same side (here: side $B$) of the chain. $(b)$: Both runners $r,r'$ have been moving towards each other until a merge can be detected and performed, locally. $r$ merges with $r^\ast$, and $r'$ merges with ${r'}^\ast$.}
    \label{fig:doubleturn}
\end{figure}
Figure~\ref{fig:doubleturn}.$(b)$ shows what happens after several steps, when both runners have moved towards each other and are close enough:
the participating robots within a constant viewing range can now detect that a merge is possible. Then $r$ merges with $r^\ast$, and $r'$ merges with ${r'}^\ast$.

For implementing this as an algorithm (Section~\ref{sec:mergeless}), we have to put some structure on the polygonal shape.
For example, the starting of runners at every vertex does not make sense if the polygon's shape includes stairways of alternating left and right turns of size one.
This leads us to the construction of \emph{modules} (cf.\ Figure~\ref{fig:stair} and \ref{fig:stairheight}).
We will introduce so called \emph{edge modules} and \emph{vertex modules} into which the chain can be decomposed and show that even if the chain intersects itself,
still two consecutive vertex modules exist such that both perform a turn in the same direction.
This then is the generalization of what we have described above, concerning Figure~\ref{fig:doubleturn}.$(a)$.
Runners will also be interpreted as modules and we call them \emph{runs}:
They are moving \emph{edge modules} of height $1$ which include the associated \emph{runner}.

Before going into more details about runs, we now present the main algorithm and the easier part, i.e., chains which allow local merges without the need of starting runs, namely the \emph{merge configurations}.
In each phase, every robot synchronously executes two procedures:
The first one continues all existing runs.
The second one every $L$ phases performs merges and starts new runs afterwards, where $L$ is a constant.
During the other phases every robot waits for the same number of timesteps, in order to achieve equally sized phases.
The definition of $L$ and some further constants will be discussed in Lemma~\ref{lem:run-properties}.

\def \algogathering{
\algorithm{gathering}
\texttt{
\begin{enumerate}
    \item execute-run() \hfill (see Algorithm~\ref{algo:move})
    \item if (phase-number $\bmod$ L = 0)
    \item $\qquad$ merge() \hfill (see Algorithm~\ref{algo:merge})
    \item $\qquad$ initialize-run() \hfill (see Algorithm~\ref{algo:runinit})
\end{enumerate}}
(For pseudocode, cf.\ Listing~\ref{lst:main-algo} in the appendix.)}%
\algogathering\label{algo:gathering}%
\medskip
\par
\noindent
We will show that after every interval of a constant number of phases, either a merge operation can be performed or else a pair of runs can be initiated, which will lead to a merge operation after at most linear many phases.
In total, this yields a linear running time for the gathering of the whole chain.

\subsection{Merge Configurations}
Looking at the shape of the chain, we say that it is a \emph{merge configuration} if for at least one robot a module of the merge actions of Table~\ref{table:robot-actions_a}.$A_{\ref{action:merge}}$ matches, i.e., the robot identifies itself as being one of the black robots of the corresponding module.
Then, the robot performs the corresponding hop action.
The result is that all these black robots perform a hop such that the subchain is shortened by two robots.
The size of the merge types is upper bounded by a constant $K$.
Later, Lemma~\ref{lem:run-properties} deals with a suitable value for this constant.

In general, merge modules can also overlap.
Figure~\ref{fig:overlap} shows how we define \emph{overlapping}:
two neighboring merge modules can either \emph{overlap} by two (Type 1) or by three robots (Type 2).
(Overlappings by only one robot are not covered by our definition because they do not need a special treatment.)
Figure~\ref{fig:overlap}.($a$) shows the Type $1$ overlapping.
Both modules overlap by two robots.
If the merge action is applied to both modules simultaneously, the robots $r$ and $r'$ swap their positions instead of merging.
If the modules overlap by three robots (Figure~\ref{fig:overlap}.($b$)), the overlap Type is $2$.
If in this case both modules would apply their merge actions at the same time, two different merge actions apply to the robot $r$.
We deal with overlapping as follows:
Looking at longest subchains of overlapping merge modules of a fixed merge type and a fixed overlap type, we want that only the outmost modules
of such subchains perform their merge action.
The robots of a merge module can detect this locally by just checking if they have an overlapping with two neighboring merge modules of their own merge type and overlap with both in the same overlap type.
We perform this check sequentially for every possible combinations of merge types and overlap types.
In case that the subchains consist of at least three modules, this restores the desired behavior of merge actions.
For subchains of two modules, the merge actions are performed as in Table~\ref{table:robot-actions_c}.$A_{\ref{action:overlap_a},\ref{action:overlap_b}}$.
In both cases, at least two robots are removed. This ensures the removal of robots when merge modules exist.
Lemma~\ref{lem:merge_working} proves that this procedure works correctly.

\begin{table}[t]
\caption{Robot actions $(c)$.}
\label{table:robot-actions_c}
\smallskip
\def\tabularxcolumn#1{m{#1}}
\centering
\begin{tabularx}{\textwidth}{p{1cm} X}
%\begin{tabularx}{0.45\textwidth}{p{0.6cm} X}
	& \hspace{2.3cm}\textbf{Pattern} \hspace{3.8cm} \textbf{Action} \\
	\hline\hline\\
	\action{action:overlap_a} & \hspace{2.3cm}\includegraphics[scale=0.85]{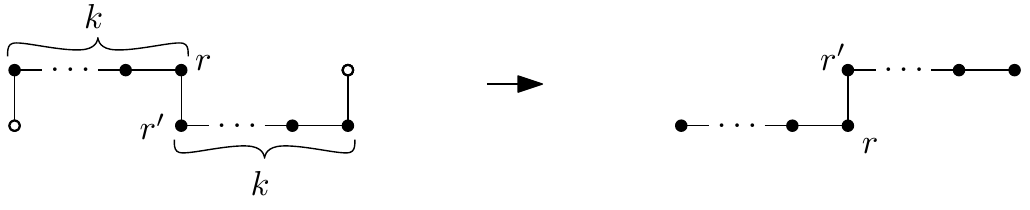} \\
	\action{action:overlap_b} & \hspace{2.3cm}\includegraphics[scale=0.85]{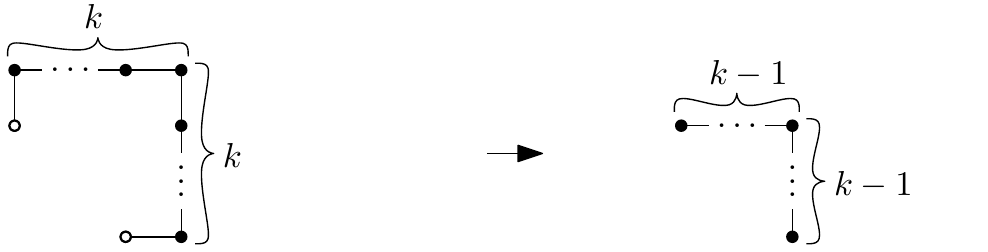} \\
	%\action{action:overlap_a} & \includegraphics[scale=0.83]{fig_spaa/action_merge_overlap.pdf} \\
    %\action{action:overlap_b} & \includegraphics[scale=0.83]{fig_spaa/action_merge_overlap2.pdf} \\
	\\
	& Symmetric hops of two overlapping merge modules of the same merge type $k$. $A_{\ref{action:overlap_a}}$: The white robots are removed and $r$ and $r'$ swap their positions. $A_{\ref{action:overlap_b}}$: Both white robots plus two of the black ones are removed by the merge. \\
	 \hline
\end{tabularx}
\end{table}

\begin{figure}[h]
\centering
        \includegraphics{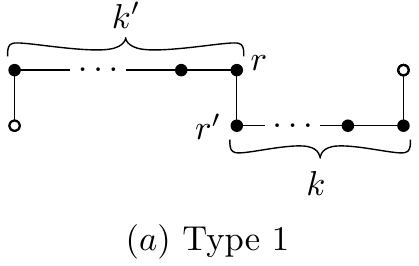}
        \hspace{1cm}
        \includegraphics{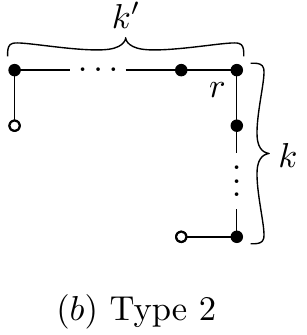}
    \caption{Definition of \emph{overlappings}. Type $1$: Two merge modules overlap by two robots. Type $2$: Two merge modules overlap by three robots.}
    \label{fig:overlap}
\end{figure}

\def \algomerge{
\algorithm{merge}
Every robot $r$ does the following:\\
for each merge type $k\in\{1,\ldots, K\}$ the robot checks:
    \begin{itemize}
    \item if merge type $k$ matches (i.e., if $r$ is one of the black robots of a merge pattern of type $k$ of Table~\ref{table:robot-actions_a}.$A_{\ref{action:merge}}$) then
        \begin{itemize}
        \item let $M'$ be this $k$ merge module of the subchain
        \item if $M'$ has an overlapping of Type $1$ with at most one other $k$ merge module, then the black robots of $M'$ perform the corresponding hop action of Table~\ref{table:robot-actions_a}.$A_{\ref{action:merge}}$ (or of Table~\ref{table:robot-actions_c}.$A_{\ref{action:overlap_a}}$ if the subchain consists of exactly two of these modules.).
        \item if $M'$ has an overlapping of Type $2$ with at most one other $k$ merge module, then the black robots of $M'$ perform the corresponding hop action of Table~\ref{table:robot-actions_a}.$A_{\ref{action:merge}}$ (or of Table~\ref{table:robot-actions_c}.$A_{\ref{action:overlap_b}}$ if the subchain consists of exactly two of these modules.).
        \end{itemize}
    \end{itemize}
All robots perform these checks fully synchronously for all $k$ and overlap types.

Finally, a call of \texttt{cleanup-runs()} solves collisions among runs (cf.\ Section~\ref{sec:mergeless}) that have come too close because of the merge operations.
(For pseudocode, cf.\ Listing~\ref{lst:merge-subphase} in the appendix.)}
\algomerge\label{algo:merge}

\medskip
\par

\begin{lemma}
\label{lem:merge_working}
If the chain is a merge configuration and algorithm \texttt{merge()} is performed, then either the gathering problem is solved or at least one robot is merged.
\end{lemma}
\begin{proof}
Obviously, every merge action merges at least one robot.
Consider a merge configuration where the execution of algorithm \texttt{merge()} could not perform any merge action.
Then, all merge modules have to be of the same merge type $k$ and have to overlap with their neighbors in the same overlap type.

If all merge types and all overlap types are equal, then the chain can only be closed if the overlap type is $2$ (Figure~\ref{fig:overlap}.($b$)).
Then, the whole chain is located within a square of at most size $(K-1)\times(K-1)$.
\end{proof}

\section{Mergeless Configurations}\label{sec:mergeless}
We call a configuration \emph{mergeless} if no merge can be performed locally, i.e., no merge module from Table~\ref{table:robot-actions_a}.$A_{\ref{action:merge}}$ applies.
To gather such configurations, we use \emph{runners}, which are robots that move along the chain and can perform diagonal hops.
We show that they can transform mergeless configurations into merge configurations so that local merges can be performed.
We define suitable starting points for runners in terms of geometric properties of the input chains which are interpreted as rectangular polygons with self-intersections.
This is necessary since the robots are indistinguishable.

For being able to put this approach into an algorithm, we do the following:
Depending on local geometric shapes of the chain, we define types of subchain modules into which a mergeless configuration can be fully decomposed.
In more detail, we define so called \emph{edge modules} (EMs) and \emph{vertex modules} (VMs) (cf.\ Figure~\ref{fig:stair}).
Both kinds of modules start and end with three collinear robots (in Figure~\ref{fig:stairheight}, these robots are filled with white color) and in between, they have a non-negative number of stairs, the so-called \emph{stairway} (cf.\ Figure~\ref{fig:stairheight} and Figure~\ref{fig:stair}).
The stairs can be seen as walking on the chain and alternatingly performing vertical and horizontal steps (or vice versa), while, concerning vertical as well as horizontal movements, we keep the initially chosen direction and never move backwards.
If the numbers of vertical and horizontal steps differ by one, the module is an EM and the lines through the first and last collinear robots are parallel.
Otherwise, the module is a VM and the lines through the first and last robots are perpendicular.
The height of a module is the maximum of the number of vertical and the number of horizontal steps in that module, whereas $\EM(h)$ ($\VM(h)$) is the set of all edge modules (vertex modules) of height $h$.
By overlapping the first (last) three robots of a module (i.e., the white ones in Figure~\ref{fig:stairheight}), these modules can be glued together.
Lemma~\ref{lem:mergelessmodules} proves that every mergeless configuration can actually be decomposed into the desired way (for the proof, see the appendix.).

\begin{figure}[t]
\centering
    \includegraphics{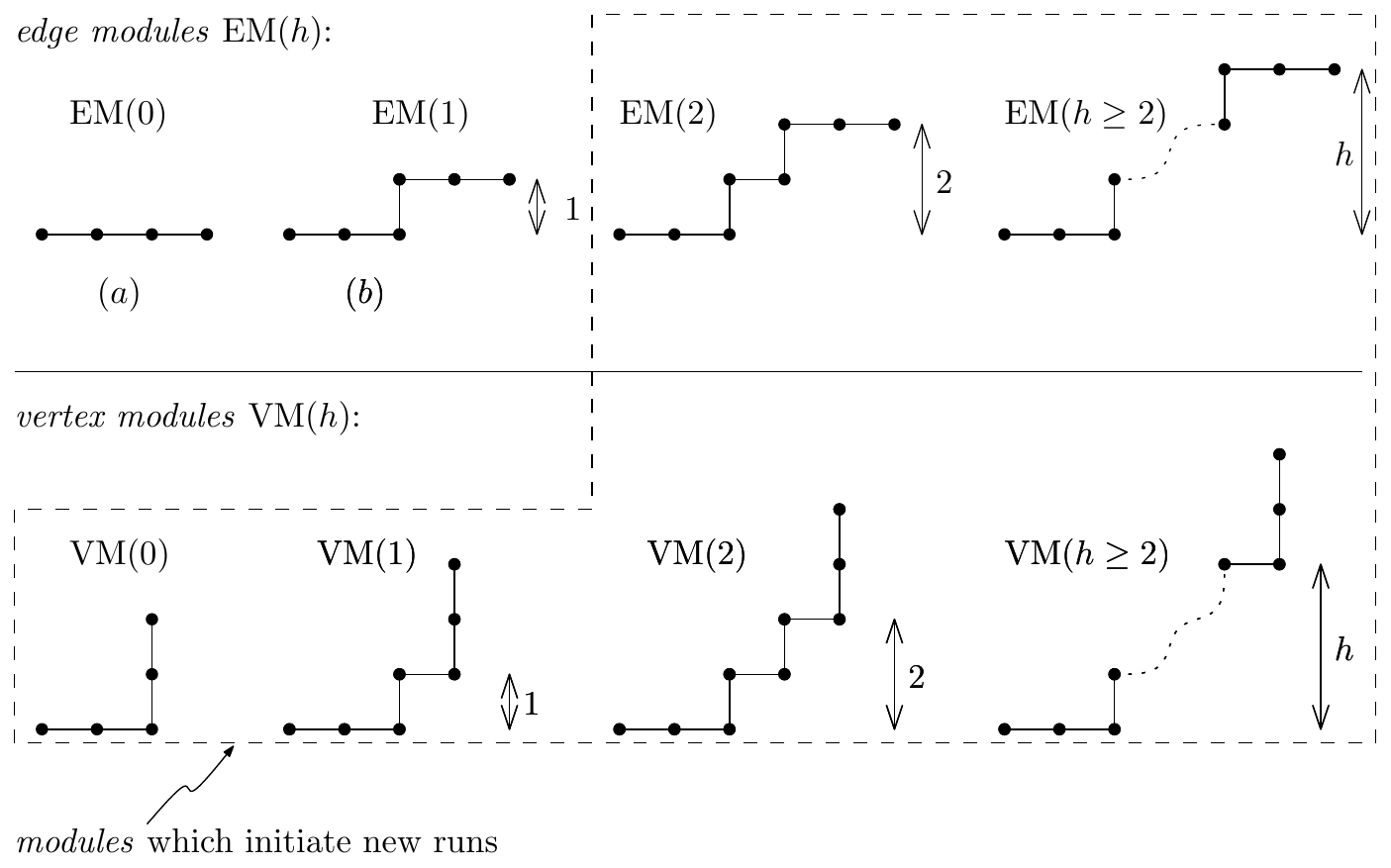}
    \caption{Edge modules and vertex modules. The modules included in the dashed polygon initiate new runs.}
    \label{fig:stair}
\end{figure}

\begin{figure}[th]
    \centering
    \includegraphics{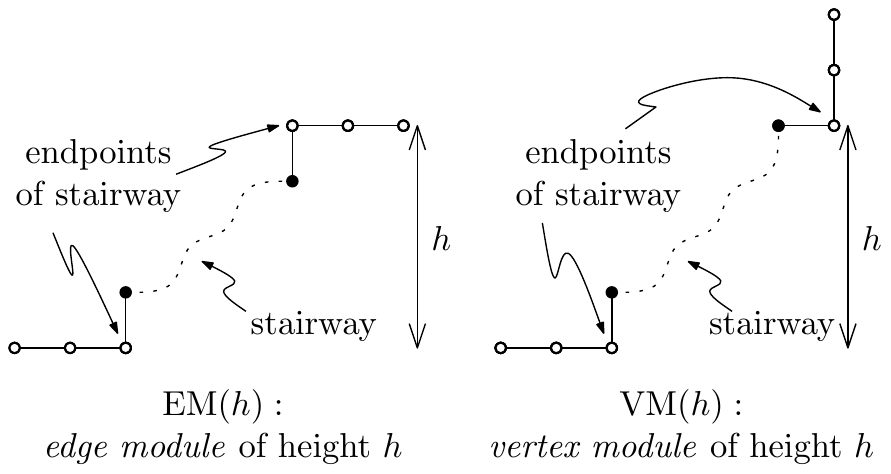}
    \caption{Height of edge modules $\EM(h)$ and vertex modules $\VM(h)$. In edge modules, the lines through the triples of white robots are in parallel. In vertex modules, these lines are perpendicular.}
    \label{fig:stairheight}
\end{figure}
\begin{lemma}
\label{lem:mergelessmodules}
Every mergeless configuration consisting of at least two robots can be fully decomposed into edge modules and vertex modules.
\end{lemma}

We also want to talk about moving \emph{modules} instead of runners.
For this, we identify runners with moving modules, namely $\EM(1)$ modules (cf.~Figure~\ref{fig:stair} and \ref{fig:goodpair})
and call these modules \emph{runs}.

We need that if we start new runs in a mergeless configuration, at least one pair of them is a ``good pair'', i.e., it modifies the chain such that a local merge will be possible.
For this, we need the geometric property of Lemma~\ref{lem:goodvertexpair} as a basis (For the proof see appendix.).
Lemma~\ref{lem:goodpairexistence} will complete this argumentation.

Consider a walk along the chain in some arbitrary but fixed chain direction.
When passing a vertex module, we either turn by $90^{\circ}$ to the left or to the right.
In the first case, we call this vertex module \emph{convex} and \emph{concave} otherwise.

\begin{lemma}\label{lem:goodvertexpair}
Every mergeless configuration contains at least one pair of consecutive vertex modules, such that both are either \emph{convex} or both are \emph{concave}.
Here, consecutive means that the subchain between them only consists of edge modules.
\end{lemma}

We let the modules contained in the dashed area in Figure~\ref{fig:stair} initialize runs at their endpoints.
The corresponding actions are specified in Table~\ref{table:robot-actions_a}.$A_{\hyperref[action:ALG_runinit_generalized_1]{\ref*{action:runinit_a}-\ref*{action:runinit_c}}}$ and $A_{\hyperref[action:ALG_runinit_generalized_2]{\ref*{action:runinit_spec_a}-\ref*{action:runinit_spec_e}}}$ (for a detailed list of all shapes considered in \texttt{initialize-run()}, see Table~\ref{table:robot-actions_b} in the appendix.).
Algorithm \texttt{initialize-run()} using these actions performs the initialization of new runs.

The EMs of height less than $2$, i.e., (a) and (b) in Figure~\ref{fig:stair}, are the only modules which do not start any runs and we call a subchain only consisting of EM(0) and EM(1) modules a \emph{quasi edge}.

\subsection{Runs}\label{ssec:runs}
Every constant number (namely $L$) of phases, algorithm \texttt{merge()} and algorithm \texttt{initialize-run()} are executed.
If the configuration is mergeless, i.e., \texttt{merge()} does not result in the merge of at least two robots, we show that then at least one pair of the initialized runs will lead to a merge after at most linear many phases.
(Recall that runs are moving $\EM(1)$ modules, including one robot as a runner. Cf.\ Figure~\ref{fig:goodpair}.)

Assume, we have a quasi edge, i.e., a subchain only consisting of $\EM(0)$ and $\EM(1)$ modules, connecting two of the run initiating modules.
At every endpoint of a quasi edge a run is initiated, such that these two runs will move towards each other.
We call these two runs a \emph{run pair}.
The runs of a pair may differ concerning reflections and rotations of their local shapes.
The run pairs we need for merges, are \emph{good pairs}.
\begin{definition}\textit{(Good Pair)}\quad
Let $R,R'$ be a \emph{run pair}, generated in a mergeless configuration.
If $R'$ is a reflection of $R$ (including the runner) over an axis which is orthogonal to the alignment of the quasi edge $\subchain{R,R'}$, then we call $R,R'$ a \emph{good pair} of runs (cf.\ Figure~\ref{fig:goodpair}).
\end{definition}
If we look at the bold marked robots in Figure~\ref{fig:goodpair}.($a$), we recognize that the subchain between these robots can be shortened according to the $L_1$-distance, but the robots cannot detect this locally.
What we mainly do is to wait until $R_0$ and $R_1$ have come close enough (by repeating execution of \texttt{execute-run()}) so that the involved robots can detect this situation within their constant viewing range.
This is valid because during their movement, the runs keep their orientation and local shape.
Then, always at least one of the merge actions (see Table~\ref{table:robot-actions_a}.$A_{\ref{action:merge}}$) can be applied.
The movement of $R_0$ and $R_1$ is assured by the actions from Table~\ref{table:robot-actions_a}.$A_{\ref{action:hop}}$ and~\ref{table:robot-actions_a}.$A_{\ref{action:hop2}}$, which can always be applied to runs on a quasi edge.

In the remaining section, we will prove the following.
\begin{enumerate}
\item If the chain is mergeless (i.e., no merge can be performed) then at least one new \emph{good pair} can be generated:
\begin{enumerate}
    \item Geometric properties (Lemma~\ref{lem:goodvertexpair}) and module decomposition (Lemma~\ref{lem:mergelessmodules}) lead to existence of a new good pair (Lemma~\ref{lem:goodpairexistence}).
\end{enumerate}
\item Pipelining (of good pairs):
\begin{enumerate}
    \item Technical properties, existence of the constants ($C,K,L$, needed by the algorithm) and correctness of collision solving in \texttt{cleanup-runs()} (see Lemma~\ref{lem:run-properties} and \ref{lem:mindist}).
    \item Good pairs on same quasi edge are initiated like pairs of enclosing parentheses (see Lemma~\ref{lem:older-runs}).
    \item Runs of a good pair keep moving towards each other until a merge happens (Lemma~\ref{lem:EM1-run} and Prop\-o\-si\-tion~\ref{prop:quasiedgedestroy}).
    \item Every good pair will have its own merge (see Lem\-ma~\ref{lem:uniquemerge}), so pipelining works as desired.
\end{enumerate}
\end{enumerate}
In Section \ref{sec:runningtime}, we will then prove that \texttt{gathering()} makes use of the above properties in such a way that we get a total linear running time.

\begin{figure}[th]
\centering
    \includegraphics{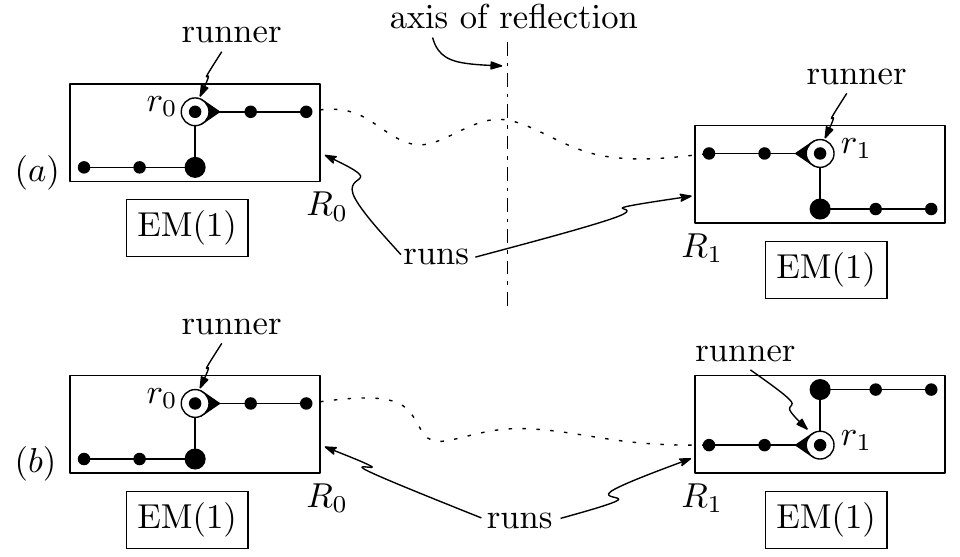}
    \caption{Run-Initialization: A \emph{good pair} of runs will lead to a guaranteed merge. $(a)$: $R_0,R_1$ is a \emph{good pair}. $(b)$: not a good pair.}
    \label{fig:goodpair}
\end{figure}

Lemma~\ref{lem:goodvertexpair} gives us for a mergeless configuration that there exists a pair of two convex or two concave vertex modules, connected only by edge modules.
If this subchain only contains EM(0) and EM(1) modules, i.e., the subchain is a quasi edge, this easily gives the existence of a good pair (cf.~the vertex modules of Figure~\ref{fig:stair} to the run init actions in Table~\ref{table:robot-actions_a}.$A_{\hyperref[action:ALG_runinit_generalized_1]{\ref*{action:runinit_a}-\ref*{action:runinit_c}}}$, $A_{\hyperref[action:ALG_runinit_generalized_2]{\ref*{action:runinit_spec_a}-\ref*{action:runinit_spec_e}}}$ (or to the detailed list in Table~\ref{table:robot-actions_b} in the appendix).

The following lemma gives also the existence if the subchain contains $\EM(h\geq2)$ modules (for the proof see appendix).
\begin{lemma}\label{lem:goodpairexistence}
For a mergeless configuration, on the subchain between two consecutive convex or two consecutive concave vertex modules, which only contains edge modules, always a good pair of runs can be generated.
\end{lemma}

{\def\arraystretch{0.6}
\begin{table}
\caption{Robot actions $(a)$.}
\label{table:robot-actions_a}
\smallskip
\def\tabularxcolumn#1{m{#1}}
\centering
%\begin{tabularx}{0.77\textwidth}{p{0.9cm} X}
\begin{tabularx}{\textwidth}{p{1cm} X}
	& \hspace{2.5cm} \textbf{Pattern} \hspace{3.5cm} \textbf{Action} \\
	\hline\hline\\
	\action{action:hop} & \hspace{2.5cm}\includegraphics[scale=0.85]{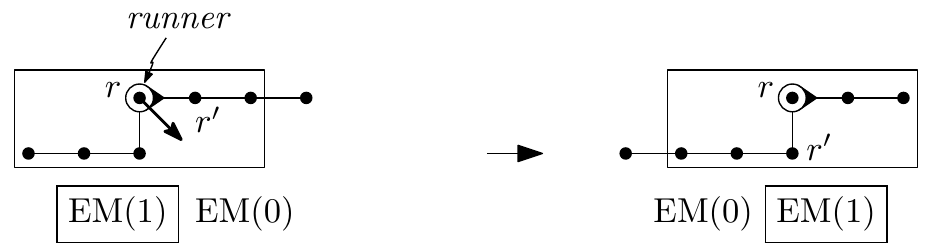} \\
	\\
	& The $\EM(1)$ module of robot $r$ moves one step further by swapping its position with the following $\EM(0)$ module. (Note that $r$ performs a diagonal hop during this movement.)\\
	\hline\\
	\action{action:hop2} & \hspace{2.5cm}\includegraphics[scale=0.85]{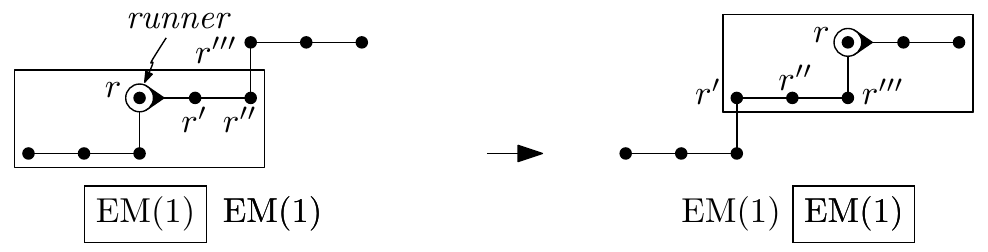} \\
	\\
	& The $\EM(1)$ module of $r$ moves three steps further by swapping its position with the following $\EM(1)$ module. \\
	\hline\\
	\action{action:merge} & \hspace{2.5cm}\includegraphics[scale=0.85]{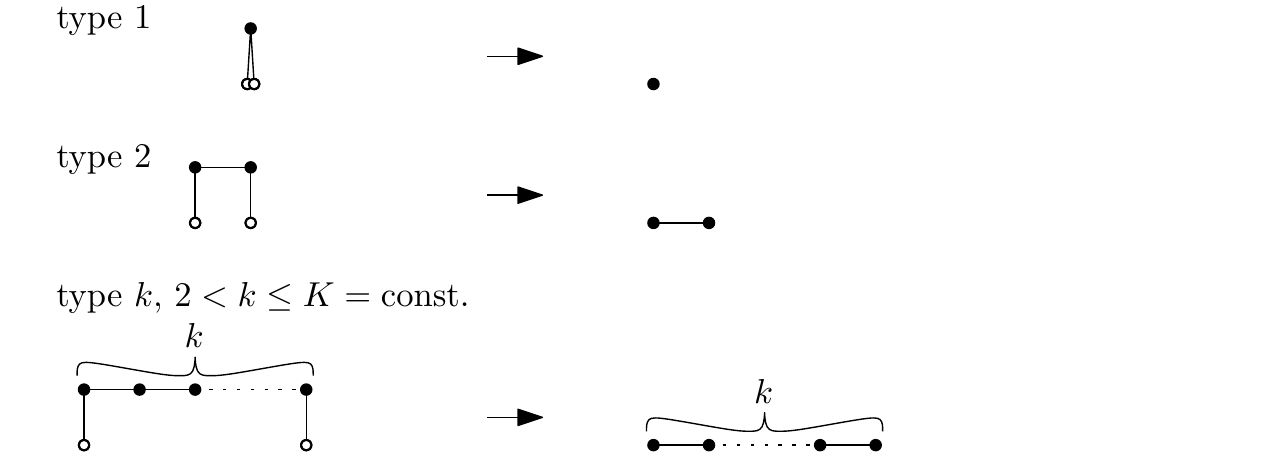} \\
	\\
	& The black robots perform a merge operation of Merge-Type $k$. I.e., they hop downwards and the outmost of them merge with the white robots (Type 1 (special case): all three robots merge to a single one). Thereby, \#robots is always decreased by $2$. (Section~\ref{ssec:runs}, \emph{runs} in \emph{mergeless configurations}: On a quasi edge, $\EM(1)$ modules are replaced by $\EM(0)$ modules.) \\
	\hline\\
	$A_{\ref*{action:runinit_a}-\ref*{action:runinit_c}\label{action:ALG_runinit_generalized_1}}$ & \hspace{2.5cm}\includegraphics[scale=0.85]{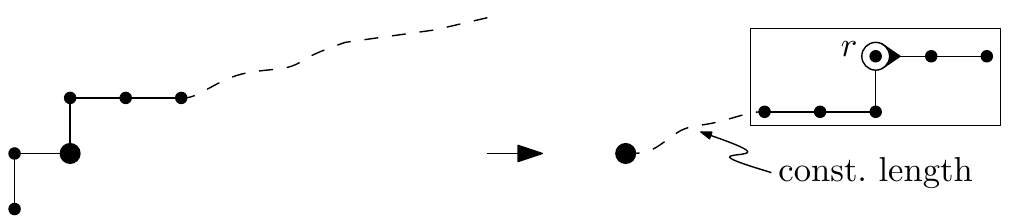} \\
	\\
	& The bold marked robot does not hop. $r$ becomes the new runner. \\
	$A_{\ref*{action:runinit_spec_a}-\ref*{action:runinit_spec_e}\label{action:ALG_runinit_generalized_2}}$ & \hspace{2.5cm}\includegraphics[scale=0.85]{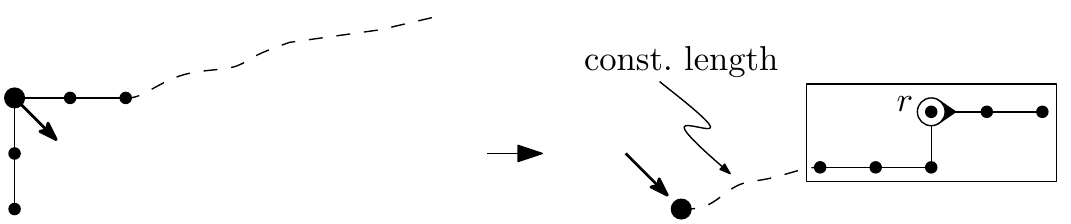} \\
	& The bold marked robot hops. $r$ becomes the new runner. \\
	\hline
\end{tabularx}
\end{table}}

\begin{figure}[h]
\centering
    \includegraphics{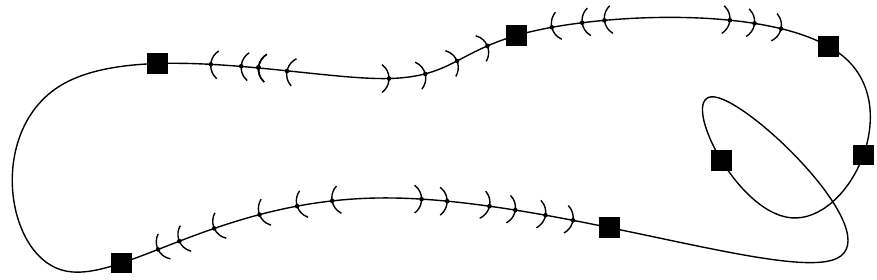}
    \caption{Intuition for pipelining of good pairs. The black boxes symbolize the starting modules of runs.
Parentheses symbolize the good pairs.}
    \label{fig:pipelining_approach}
\end{figure}

We now explain the main algorithmic part of the  gathering algorithm.
In Lemma~\ref{lem:goodpairexistence}, we have shown that good pairs of runs always exist.
In order to achieve a linear running time, we need that every such pair will lead to a merge operation and that every good pair has its own one, i.e., it is
uniquely associable.
This can happen before both partners meet each other (then only one of them is part of the merge) or at the latest when they meet.
Furthermore, we need, that all this also works if we start new good pairs after every constant number of phases.
Technically, this means that the partners of every good pair can move towards each other without colliding with other good pairs that have not yet led to a merge.
This can be assured, if new good pairs fully include older good pairs (or if the subchains between them are disjoint).
I.e., they behave like enclosing parentheses, as shown in Figure~\ref{fig:pipelining_approach}.
We call the whole concept \emph{pipelining}.
In Lemma~\ref{lem:older-runs}, we prove that good pairs are actually generated like enclosing pairs of parentheses.
The technical preliminaries for the correct working of the pipelining are proven by Lemma~\ref{lem:run-properties}.

Next, we deal with some more details about the runs:
Roughly speaking, runs are started at the endpoints of the stairways $\EM(h>1)$ and $\VM(\cdot)$ modules (cf.\ Figure~\ref{fig:stairheight}).
Runs are associated with moving $\EM(1)$ modules.
If a new run is initiated, such a module must be created (if needed, by a local modification of the chain) at the beginning of the quasi edge.
This modification needs some caution, because we do not want to destroy the quasi edge structure (recall that quasi edges can consist of $\EM(0)$ \emph{and} $\EM(1)$ modules.).
Hence, the new run may have to be started some constant distance apart from the beginning of the quasi edge.
The resulting set of actions can be found in Table~\ref{table:robot-actions_a} and more detailed in Table~\ref{table:robot-actions_b} in the appendix:
$A_{\hyperref[action:ALG_runinit_generalized_1]{\ref*{action:runinit_a}-\ref*{action:runinit_c}}}$ for all $\VM(h>0)$ and $\EM(h\geq 2)$, and $A_{\hyperref[action:ALG_runinit_generalized_2]{\ref*{action:runinit_spec_a}-\ref*{action:runinit_spec_e}}}$ for all $\VM(0)$ modules.
The latter requires more cases because else the two runs, simultaneously generated by an $\VM(0)$ module, are started too close to each other.

Each run $R$ carries a timestamp, which stores the phase number of the time when $R$ was started.
This timestamp ensures that runs of good pairs that did not have a merge yet, will survive if they collide with other runs.

During the run initialization it might happen that newly initiated runs overlap with other runs.
In order to prohibit such situations, this is solved immediately after the initialization and prior to the call of \texttt{cleanup-runs()} of the \texttt{initialize-run()} algorithm.

\def \algoinitialize{
\algorithm{initialize-run}
Every robot $r$ does the following:
\begin{itemize}
\item top down checks for all initialization patterns $A_{\hyperref[action:ALG_runinit_generalized_1]{\ref*{action:runinit_a}-\ref*{action:runinit_c}}}$ and $A_{\hyperref[action:ALG_runinit_generalized_2]{\ref*{action:runinit_spec_a}-\ref*{action:runinit_spec_e}}}$ (of Table~\ref{table:robot-actions_a} and \ref{table:robot-actions_b} (appendix), respectively) if it is part of such a pattern and takes the first one that matches. (If none matches then it does nothing.)
\item if necessary, it performs the corresponding hop.
\item if $r$ is the new runner, then it becomes active and stores a current timestamp from the global clock.
\end{itemize}
Then, solve overlappings with other runs: if $r$ is runner of some run $R$, then
\begin{itemize}
\item if $R$ overlaps with a run of same age, then both are stopped.
\item if $R$ overlaps with a run of different age, then the older one is stopped.
\end{itemize}
Finally, a call of \texttt{cleanup-runs()} solves the remaining collisions
(cf.\ Listing~\ref{lst:run-init-subphase} in the appendix.)}
\algoinitialize\label{algo:runinit}
\medskip
\par

The runs move along the chain as follows (cf.\ Algorithm~\ref{algo:move} \texttt{execute-run()}).
As runs are moving $\EM(1)$ modules on a quasi edge and the latter consists only of $\EM(h\leq 1)$ modules, the movement is performed by swapping the position with the next
$\EM(h\leq 1)$ module in moving direction.
This is done as follows:
If the chain looks like the pattern of Action~$A_{\ref{action:hop}}$ (of Table~\ref{table:robot-actions_a}), then the runner first initiates a diagonal hop and then moves one robot further along the chain.
If the chain looks like in the pattern of Action~$A_{\ref{action:hop2}}$, then the runner moves three robots further along the chain (without any diagonal hops) and will wait there for the two following calls of \texttt{execute-run()} for balancing the fluctuations in speed.
In three consecutive calls of \texttt{execute-run()}, each run moves at least three robots further on the chain.
As a result, the runs of a good pair are also moving closer together.

\def \algoexecute{
\algorithm{execute-run}
Every run $R$ does the following:\\
If not pausing (see below), then:
\begin{itemize}
\item If the pattern of Action~$A_{\ref{action:hop}}$ (of Table~\ref{table:robot-actions_a}) matches, then a hop is performed and $R$ moves one robot further.
\item If the pattern of Action~$A_{\ref{action:hop2}}$ (of Table~\ref{table:robot-actions_a}) matches, then $R$ moves three robots further.
(The additional two steps will be balanced by pausing (i.e., not moving) during the following two calls of \texttt{execute-run()}.)
\end{itemize}
For implementing the pausing, every robot decreases a local counter from value $2$ to $0$.
At the end, the \texttt{cleanup-runs()}-algorithm is executed (cf.\ Listing~\ref{lst:module-movement-subphase} in the appendix.).}
\algoexecute\label{algo:move}
\medskip\par

From now on, we say that two runs $R,R'$ are \emph{colliding}, if $\dist{R,R'}\leq C=\mathrm{const.}$ (for details of $C$ see Lemma~\ref{lem:run-properties}.).
In order to prevent new good pairs from colliding with the enclosed ones, we initiate new runs only after every constant number of $L$ phases.
But the other way around, enclosed runs can move towards enclosing good pairs.
For example, in case that only one run of a good pair initiated a merge (and was stopped/removed by it), its partner run is still active and 
can collide with a run of the enclosing good pairs.
The single run has to be stopped.
According to Lemma~\ref{lem:older-runs}, this run is older than the colliding good pair.
This is handled in \texttt{cleanup-runs()} of Algorithm~\ref{algo:cleanup}:
if two runs of different ages collide, it stops the older one.

The \texttt{cleanup-runs()}-algorithm also stops colliding runs which are of the same age if they are not a good pair.
Runs of good pairs are stopped during the merge action (case \ref{enum:cleanup-merge}.\ of \texttt{cleanup-runs()}).
(Note that in case of collision the runs are close enough in order to locally detect whether they are a good pair or not.)
If both partners of a good pair are close enough for participating in a merge, the constant $K$ (cf.\ Lemma~\ref{lem:run-properties}) ensures that then this merge can actually be performed and the participating runs are stopped.
The merge replaces the participating $\EM(1)$ modules by $\EM(0)$ modules.
(Note that in contrast to the merges we have dealt with in the section about merge configurations, this kind of merge does not cause symmetry issues.)
The existence of suitable constant values for $C,K$ and $L$ is proven in Lemma~\ref{lem:run-properties}.
 
\def \algocleanup{
\algorithm{cleanup-runs}
Every run $R$ performs the following:
\begin{enumerate}
\item solve collisions with too close runs $R':$ $R'$ is located in moving direction of $R$ (on the same quasi edge):\label{enum:cleanup-colliding}
\begin{itemize}
\item if both are of different ages, then stop the older one
\item if both are of the same age, then
    \begin{itemize}
    \item if both are moving in the same direction, then stop $R$ (the rear one)
    \item else if both are not a good pair, then stop both (good pairs are handled below)
    \end{itemize}
\end{itemize}
\item perform merges, initiated by good pairs (and stop the participating runs):\label{enum:cleanup-merge}
\begin{itemize}
\item if $R$ is part of a merge module (i.e., one of the merge types $k$ in Table~\ref{table:robot-actions_a}.$A_{\ref{action:merge}}$, then apply the corresponding merge action and stop participating runs
\end{itemize} 
\end{enumerate}
(cf.\ Listing~\ref{lst:cleanup} in the appendix.)
}\algocleanup\label{algo:cleanup}\medskip\par
In this section, our goal is the proof of Lemma~\ref{lem:uniquemerge}.
It proves that the pipelining works as desired---i.e., different good pairs lead to different merges---, which is the basis for the achievement of the linear total running time.

On the basis of the suitable constants $C,K,L$ Lemma~\ref{lem:run-properties} 
also gives some technical properties of runs, which are needed in some of the proofs of the next lemmas (For the proof see appendix.).
\begin{lemma}\label{lem:run-properties}
There exist constant values of $C,L$ and $K$ such that the following holds during the execution of the algorithm:
\begin{enumerate}[(a)]
	\item Runs are node-disjoint.\label{enumitem:disjoint-runs}
    \item Two consecutive runs on the same quasi edge that are running in the same direction are initiated this way that they do not collide.\label{enumitem:consec-runs-collisions}
	\item Two consecutive runs that are running in the same direction cannot be removed by the same merge.\label{enumitem:consec-runs-merges}
	\item A run on a quasi edge is not merged from behind (in terms of its moving direction).\label{enumitem:behind-run-merges}
\end{enumerate}
\end{lemma}
\begin{proposition}\label{prop:quasiedgedestroy}
Neither merges nor run executions can locally destroy the quasi edge property.
\end{proposition}
In order to ensure that good pairs lead to a merge, the first thing we have to ensure is that \texttt{cleanup-runs()} does not stop them before.
This is proven in the following two lemmas.
Lemma~\ref{lem:older-runs} additionally shows that good pairs are actually generated like enclosing pairs of parentheses as symbolically depicted in Figure~\ref{fig:pipelining_approach}
(for proof of Lemma~\ref{lem:older-runs} and \ref{lem:mindist}, see appendix.).

\begin{lemma}\label{lem:older-runs}
Runs between a good pair $P=\{R,R'\}$ are older than $P$.
In particular, if $P'=\{S,S'\}$ is a good pair which is older than $P$, $C \coloneqq \subchain{R,R'}$ and $C' \coloneqq \subchain{S,S'}$, then
$C$ and $C'$ can only intersect such that either $C\cap C' = \emptyset$ or $C\cap C' = C'$.
\end{lemma}

\begin{lemma}\label{lem:mindist}
The algorithm \texttt{cleanup-runs()} does not stop any run of a good pair before one of them is part of a merge.
\end{lemma}

From now on, we can assume that all good pairs fulfill the properties as stated in Lemma~\ref{lem:run-properties}.
When introducing \texttt{execute-run()}, we have already argued that it takes at most three executions until \texttt{execute-run()} has moved the partners of a good pair closer to each other.
Using this, the technical run properties of Lemma~\ref{lem:run-properties} and the correct working of \texttt{cleanup-runs()} (Lemma~\ref{lem:mindist}),
we can prove that our algorithm ensures that every good pair $R,R'$ leads to a merge and that this merge will happen on the $\subchain{R,R'}$.
(The latter is important for ensuring that this merge does not also belong to a different good pair (Lem\-ma~\ref{lem:uniquemerge}).)
\begin{lemma}\label{lem:EM1-run}
The two runs $R,R'$ of a good pair move towards each other until at least one of them participates in a merge, which is completely included in the $\subchain{R,R'}$.
\end{lemma}
(For the proof of Lemma~\ref{lem:EM1-run}, see the appendix.)

Speeding up the gathering of the chain to a linear running time, we need that different good pairs lead to different merges.
This is proven by Lemma~\ref{lem:uniquemerge}.

\begin{figure}[t]
\centering
    \includegraphics{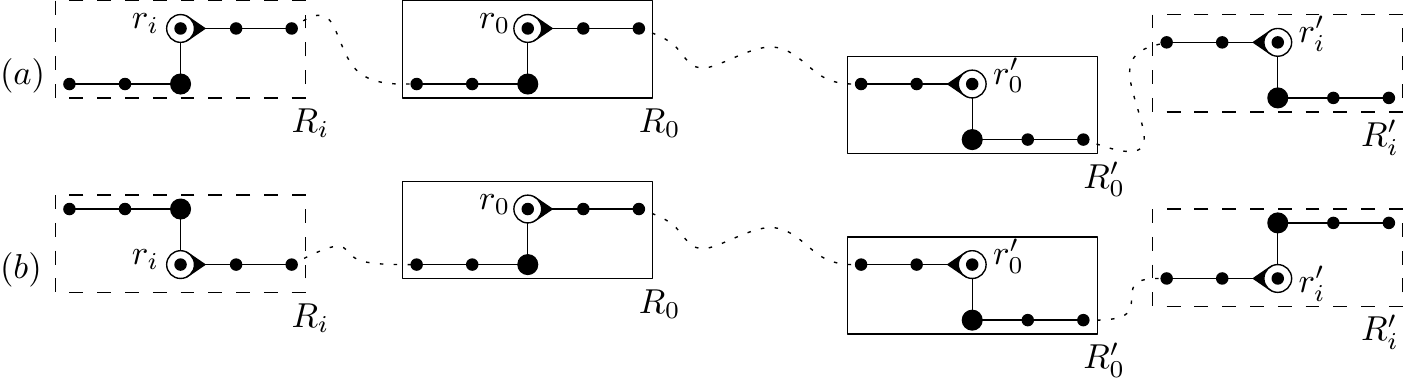}
    \caption{Every pipelined good pair has its own merge.}
    \label{fig:manhattanmerge}
\end{figure}
\begin{lemma}\label{lem:uniquemerge}
Every good pair $R,R'$ leads to a merge.
This merge can be uniquely associated with this pair and is completely included in the $\subchain{R,R'}$.
\end{lemma}
\begin{proof}
Cf.\ Figure~\ref{fig:manhattanmerge}.$(a)$.
We assume that the $\EM(1)$ modules in this figure are located on the same arbitrary quasi edge and consider the innermost good pair $R_0,R_0'$.
Because of Lemma~\ref{lem:EM1-run}, this pair will lead to a merge operation completely included in the $\subchain{R_0,R_0'}$.
We can associate this merge with $R_0,R_0'$, since by Lemma~\ref{lem:run-properties}.(\ref{enumitem:consec-runs-merges}) none of the outer runs can be involved.

We now inductively continue this idea:
Let $R_i,R_i'$ be one of the other good pairs on the current quasi edge, which has been initiated later and completely includes $R_0,R_0'$.
We assume $R_{i-1},R_{i-1}'$ being the outermost good pair for which holds $\subchain{R_{i-1},R_{i-1}'}\subset\subchain{R_i,R_i'}$.
(\OBdA we assume that $R_i$ is located on the same side as $R_{i-1}$.)
As the induction hypothesis we assume that the lemma holds for $R_{i-1},R_{i-1}'$.
Because of Lemma~\ref{lem:EM1-run}, $R_i$ or $R_i'$ will be part of a merge, completely included in the $\subchain{R_i,R_i'}$.
If this happens, Lemma~\ref{lem:run-properties}.(\ref{enumitem:consec-runs-merges}) ensures that none of the runs of the included good pairs can be also part of it.
Hence, this merge can be uniquely associated with $R_i,R_i'$.
\end{proof}

\section{Running Time}\label{sec:runningtime}
Now, we prove that the main Algorithm \ref{algo:gathering} \texttt{gathering()} makes use of the above properties in such a
way that we get a total linear running time.

\begin{lemma}\label{lem:nomergethengoodpair}
There exists a constant $m$, such that it holds for every chunk of $m$ phases:
If no merge can be performed, then at least one new good pair of runs can be started or otherwise gathering is finished.
\end{lemma}
\begin{proof}
Every constant number (namely $L$) of phases, our algorithm first executes the \texttt{merge()}, immediately followed by an execution of \texttt{initialize-run()}.
If in \texttt{merge()} no merge is performed, then either gathering is finished or the configuration is mergeless.
In the latter case, Lemma~\ref{lem:goodvertexpair} and \ref{lem:goodpairexistence} ensure that then in \texttt{initialize-run()} a good pair is initiated.
\end{proof}

Because \texttt{execute-run()} is executed during every phase, it takes at most a linear number of timesteps until the runs of a good pair meet each other.
So, together with Lemma~\ref{lem:uniquemerge} (I.e., different good pairs lead to different merges.), we get to the final theorem.

\begin{theorem}\label{thm:totalrunningtime}
The distributed Algorithm~\ref{algo:gathering} \texttt{gathering()} solves the gathering problem in time $\calO(n)$.
This time is asymptotically optimal.
\end{theorem}

Note that the lower bound is given by the diameter of the initial configuration.

\section{Outlook}
As mentioned earlier, closed chains typically are worst case examples for gathering algorithms in synchronous models with local view.
Therefore, our result in this paper leads us to conjecture that in such models gathering is possible in linear time, for arbitrary connected configurations on the grid and in the Euclidean plane.
%\clearpage
%\bibliographystyle{abbrv}
%\bibliography{references}

\begin{sloppy}
\printbibliography
\end{sloppy}

%\newpage
\clearpage
\appendix
\section{Appendix}
\setcounter{algorithmNumber}{0}
\subsection{Algorithm Pseudocode}\label{appssec:pseudocode}
\lstset{frame=single}
\lstset{postbreak=\raisebox{0ex}[0ex][0ex]{\ensuremath{\hookrightarrow\space}}}
\lstset{breaklines=true,breakatwhitespace=true}

\begin{lstlisting}[label={lst:main-algo},mathescape,caption=Main Algorithm: \texttt{gathering()}.]
$p \leftarrow p+1$                   // models global clock
execute-run()      // see Listing $\ref{lst:module-movement-subphase}$
if $p \mod L = 0$:
  merge()          // see Listing $\ref{lst:merge-subphase}$
  initialize-run() // see Listing $\ref{lst:run-init-subphase}$
\end{lstlisting}

\begin{lstlisting}[label={lst:merge-subphase},mathescape,caption=Merge-Subphase: \texttt{merge()}.]
for all $k\in\{1,\ldots,K\}$
  if $r$ is one of the black robots of a merge module of type $k$ (see pattern $A_{\ref{action:merge}}$)
    let $M'$ be this module
    if $M'$ has a Type $1$ overlapping with $\leq 1$ many merge modules of Type $k$
       apply the corresponding action of $A_{\ref{action:merge}}$ and $A_{\ref{action:overlap_a}}$, respectively
    if $M'$ has a Type $2$ overlapping with $\leq 1$ many merge modules of Type $k$
       apply the corresponding action of $A_{\ref{action:merge}}$ and $A_{\ref{action:overlap_b}}$, respectively
cleanup-runs()     // see Listing $\ref{lst:cleanup}$
\end{lstlisting}
\begin{lstlisting}[label={lst:run-init-subphase},mathescape,caption=Run-Initialization-Subphase: \texttt{initialize-run()}.]
for $A_i$ in $[A_{\ref{action:runinit_a}},\ldots,A_{\ref{action:runinit_c}},A_{\ref{action:runinit_spec_a}},\ldots,A_{\ref{action:runinit_spec_e}}]$:
  if pattern $A_i$ matches:
    apply action $A_i$, start run $R$ with $\timestamp{R}=p$
    break
if $r$ is runner of a run $R$ and $\exists$run $R'$ with $R'\cap R\neq\emptyset$ and $\timestamp{R}\leq\timestamp{R'}$: // solve overlappings
  stop run $R$
cleanup-runs()     // see Listing $\ref{lst:cleanup}$
\end{lstlisting}
\begin{lstlisting}[label={lst:module-movement-subphase},mathescape,caption=Run-Execution-Subphase: \texttt{execute-run()}.]
$\mathrm{pause} \leftarrow \max\{0,\mathrm{pause}-1\}$:
if $\mathrm{pause}=0$
  if pattern $A_{\ref{action:hop}}$ matches:
    apply action $A_{\ref{action:hop}}$
  else if pattern $A_{\ref{action:hop2}}$ matches:
    apply action $A_{\ref{action:hop2}}$
    $\mathrm{pause} \leftarrow 3$
cleanup-runs()     // see Listing $\ref{lst:cleanup}$
\end{lstlisting}

\newpage
\begin{lstlisting}[label={lst:cleanup},mathescape,caption=Cleanup: \texttt{cleanup-runs()}.]
each run $R$ checks:
  if $\exists$run $R'$ in moving direction on same quasi edge as $R$ with $\dist{R,R'}\leq C$:
    if $\timestamp{R}\neq\timestamp{R'}$:
      if $\timestamp{R}<\timestamp{R'}$: stop $R$
      else: stop $R'$
    else
      if $R,R'$ are moving in the same direction: stop $R$
      else if not ($(R,R')$ is good pair): stop $R$  /* $R'$ will also stop */
for each merge type $k$:
  if pattern $A_{\ref{action:merge}}$ matches locally for merge type $k$:
    if $\exists$runner on a black node in pattern $A_{\ref{action:merge}}$:
      apply action $A_{\ref{action:merge}}$ and stop run/runs
\end{lstlisting}

\subsection{Proofs}
\begin{proof}[Lemma~\ref{lem:mergelessmodules}]
First, we fix an arbitrary robot of the chain and follow the chain from there (in arbitrary direction).
On this walk, no step backwards is possible without getting a merge of type 1 (cf.\ Table~\ref{table:robot-actions_a}.$A_{\ref{action:merge}}$) and thus, we only have to consider steps forward, to the left, and to the right.
Further, no two consecutive turns can turn in the same direction without getting a merge of type 2.
This gives, when following the chain, there must be alternating left- and right turns or a turn must be following by a forward step (resulting in three robots on a line).
Since the chain is connected and consists of more than one robot, there must be at least two left- or two right turns that are connected by forward steps.
Hence, the chain contains a sequence of at least three robots on a line.

We now consider a longest sequence $L$ of robots on a straight line and the consecutive subchain $S$, until there are again three robots on a line.
As argued before, $S$ must consist of alternating left- and right-turns.
Depending on whether the number of turns is even or odd, this is either a vertex or an edge module, and can be completely covered.
The subchain $L$ can be completely covered by overlapping $\EM(0)$ modules.
By repeating this argument, we get that the whole chain can be covered by edge and vertex modules.
\end{proof}

\begin{proof}[Lemma~\ref{lem:goodvertexpair}]
Let $s$ and $s'$ be the starting points of two consecutive vertex modules, whereas the first one is convex and the second one is concave.
To ease the discussion, we consider a vertex module consisting of a vertical vector $\vec{v}$, a diagonal vector $\vec{v}^\ast$, and a horizontal vector $\vec{v}'$ (cf.\ Figure~\ref{fig:goodvertexpair-existence}).
Here, both $\vec{v}$ and $\vec{v}'$ have length $1$ and the $L_1$-length of $\vec{v}^\ast$ is the height of the vertex module.
We say, the \emph{orientation vector} of the vertex module is the addition of these three vectors.

For these two vertex modules, we define two half-planes $H$ and $H'$.
The half-plane $H$ originates at the successor of $s$, contains $s$, and is orthogonal to the orientation vector of the first vertex module.
For the second half-plane $H'$, we construct it analogously to the first one but to contain $s'$ and to be orthogonal to the orientation vector of the second vertex module.
Since one vertex module is convex and the other one concave, and in between there are only edge modules, the orientation vectors of both vertex modules are parallel.
The $L_1$-distance between $s$ and $s'$ is at least $2$ and every edge module in between $s$ and $s'$ increases their minimal distance by at least one. Hence, $H$ is a proper subset of $H'$.

We now walk along a closed chain in an arbitrary direction, starting from an arbitrary robot.
Using the arguments from the proof of Lemma~\ref{lem:mergelessmodules}, there exist three robots on a line.
Let $x$ be the first one of them.
Assuming that all the following modules are edge modules, every of these modules increases the $L_1$-distance to $x$ by at least one.
So there has to exist at least one vertex module.
Let $s$ of Figure~\ref{fig:goodvertexpair-existence} be its starting point.
Using only edge modules, we cannot reenter the interior of $H$.
Since the chain is closed, there must be at least a second vertex module in the chain.

Using our previous discussion, for any consecutive pair of alternating convex and concave vertex modules, the half-plane of the first module is contained in the half-plane of the second one.
Hence, $x$ is contained in both half-planes and the distance to $x$ increases.
Since the chain is connected, this gives a contradiction.
\end{proof}

\begin{figure}[t]
\centering
    \includegraphics{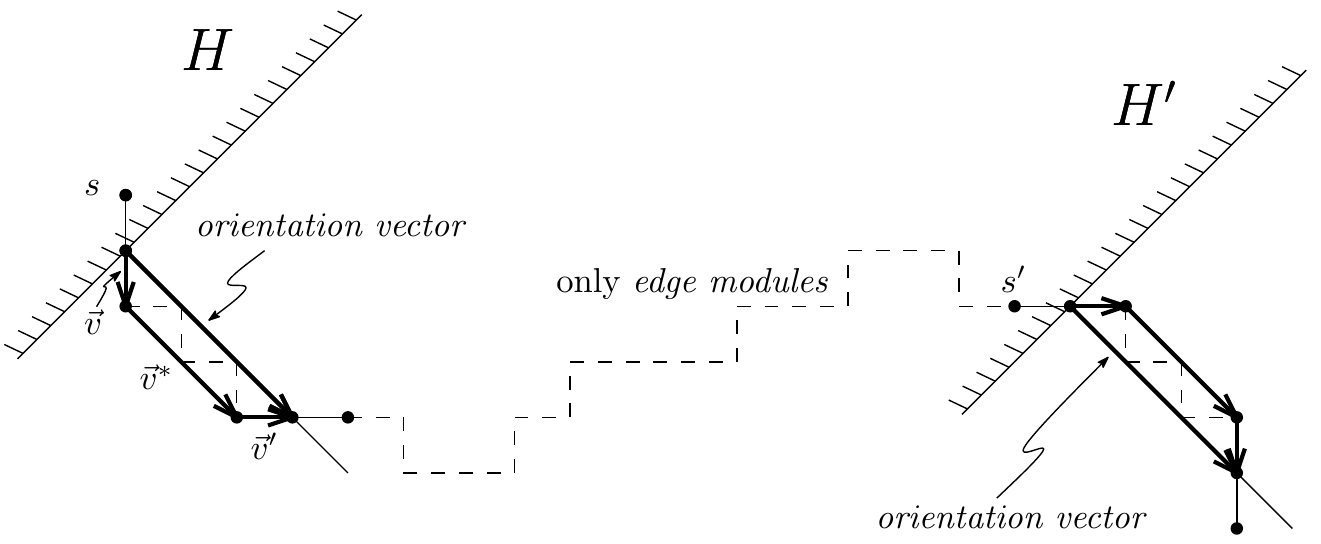}
    \caption{If convex and concave vertex modules are alternating, the starting point $s$ can never be reached again, i.e., the chain cannot be closed.}
    \label{fig:goodvertexpair-existence}
\end{figure}

\begin{proof}[Lemma~\ref{lem:goodpairexistence}]
We consider two consecutive convex or two consecutive concave vertex modules $\VM_1$ and $\VM_2$.
Between them, edge modules $\EM(h)$ with $h\geq 2$ also generate runs.
These initiate hops to different sides of the chain.
Denote by $i$ the run generated at $\VM_1$ moving towards $\VM_2$ and by $i+j$ the $j^{th}$ run generated after $\VM_1$ before $\VM_2$.
Now, if run $i+1$, generated by the first edge module after $\VM_1$, initiates hops to the same side as $i$, we have found a good pair.
If not, then the following run $i+2$ at the edge module must hop to the same side as run $i$.
Then, we can apply the same argument inductively until we find a good pair of runs or reach $\VM_2$:
Either run $i+2k+1$ generated at the next edge module hops to the same side as $i+2k$ or $i+2(k+1)$ hops to the same side as $i+2k$.
If we have to iterate this until reaching $\VM_2$, the property, $\VM_1$ and $\VM_2$ are either both convex or both concave, satisfies the lemma:
The run $i$ which was generated by $\VM_1$ initiates hops to the same side as the run, generated by $\VM_2$ and moving towards $\VM_1$.
This then finishes the proof.
\end{proof}

\begin{proof}[Lemma~\ref{lem:run-properties}]
(\ref{enumitem:disjoint-runs}):
If two consecutive runs $R,R'$ are moving in the same direction then we can choose the constant $C$ large enough such that $\dist{R,R'}<1$ does not happen.
If they are moving towards each other and are not a good pair, then $C$ also ensures that the overlapping can be detected.
And if they are a good pair, then the constant $K$ can be chosen large enough to enable their merge early enough
(Note that in \texttt{cleanup-runs()} a good pair is not interpreted as colliding and so in the property of (\ref{enumitem:disjoint-runs}), the choice of $K$ does not depend on the value of $C$.).
When new runs are initialized, they can overlap with other runs. Algorithm \ref{algo:runinit} \texttt{initialize-run()} immediately solves these overlappings after the init.\\
(\ref{enumitem:consec-runs-collisions}):
After initialization of a run on a quasi edge, this run leaves a quasi edge behind (cf.\ Proposition~\ref{prop:quasiedgedestroy}).
Before the next run is initialized, this quasi edge can be shortened by one by a merge.
Additionally, depending on which initialization pattern of Table~\ref{table:robot-actions_b} (appendix) applies, the new run is already started some constant number of
robots away from the beginning of the quasi edge.
In order to prohibit that the new run collides with its predecessor (i.e. their distance becomes $\leq C$), we need to wait long enough (cf.\ constant $L$) between the generation of two consecutive runs.
The distance between two runs, moving in the same direction, can, because of the pausing in Algorithm \ref{algo:move} \texttt{ex\-e\-cute-run()} fluctuate by four.
This also requires to chose $L$ large enough.
Both is ensured by suitably choosing the value of the constant $L$ such that $L>C+\mathrm{const}$.
On a quasi edge, a merge between consecutive runs cannot happen without removing at least one of them.
Also, if the middle one of three consecutive runs is removed, this cannot lead to a collision.
So merges cannot cause a contradiction of (\ref{enumitem:consec-runs-collisions}).\\
(\ref{enumitem:consec-runs-merges}):
This point is ensured by choosing $C$ large enough such that $C+4>K-1$ holds.\\
(\ref{enumitem:behind-run-merges}):
To ensure this, the outer runs on a quasi edge must not be merged during Algorithm \ref{algo:merge} \texttt{merge()}.
This can be prohibited by calling Algorithm \ref{algo:move} \texttt{execute-run()} sufficiently many times after \texttt{initialize-run()} before \texttt{merge()} such that
the value of $K$ is not large enough anymore for performing a merge behind the run.
So we get for the value of $L$ that $L>K+\mathrm{const.}$\ must hold.
\end{proof}

\begin{proof}[Lemma~\ref{lem:older-runs}]
The subchain $C$ between the partners of $P=\{R,R'\}$ is by definition a \emph{quasi edge}.
At the time when $P$ is initiated, all existing runs on this quasi edge are older than $P$.
Thus, if a newer or coeval run $R^\star$ exists between $R$ and $R'$ at any later time, it would have to be created after or at the same time as the start of $P$.
Because runs are all moving at the same speed and stop if they are moving too close towards each other, $R^\star$ would have to be created
on the subchain between $R$ and $R'$.
This would require the existence and creation of an $\EM(h\geq 2)$ or $\VM(\cdot)$ module on the quasi edge, contradicting Proposition~\ref{prop:quasiedgedestroy}.

Secondly, we assume that only one of the partners $S,S'$, say $S'$, is located on $C$.
This then leads us to the order $S \ldots R \ldots S' \ldots R'$.
Then, $R$ must be located on $C'$ and $R'$ is outside of the subchain.
Then, either $R$ is younger than $S,S'$ or else $S'$ is younger than $R,R'$, contradicting the first part of the lemma.
\end{proof}

\begin{proof*}[Lemma~\ref{lem:mindist}]
Assume $R$ and $R'$ are two consecutive runs on the same quasi edge.
\OBdA assume that this quasi edge is horizontally aligned and $R$ is located to the right of $R'$.
At first, we consider the case that $R,R'$ are moving towards each other.
In this case, \texttt{cleanup-runs()} stops the older run.
Now, consider $R$ to be older than $R'$ and to belong to a good pair that did not lead to a merge yet.
But in this case, the partner of $R$, which we call $R^{\star}$, must be located left of $R'$, contradicting Lemma~\ref{lem:older-runs}.
Otherwise, with $R$ and $R'$ having the same age, we are allowed to delete both if they are not a good pair (else we will get a merge.).

Now, assume that $R$ and $R'$ are moving in the same direction along the chain and that $R$ is following $R'$.
It remains to distinguish the following cases:
\begin{enumerate}
\item $R$ is older than $R'$ and $R$ did not yet have an associated merge:
    Then, the partner $R^{\star}$ of $R$ would have to be located on the left of $R'$.
    This would contradict Lemma \ref{lem:older-runs}.
    Hence, $R,R^{\star}$ cannot be a good pair and we are allowed to stop the older run.
\item $R$ is younger than $R'$:
    Then, not both can belong to good pairs because else the merge which would have shortened the subchain between both too much, would have removed one of them.
    So, assume that just $R'$ is part of a good pair.
    After its initiation, $R'$ starts leaving a quasi edge behind itself, and the constants (cf.\ Lemma~\ref{lem:run-properties}) in conjunction with our algorithm ensure that no run after $R'$ can be started early enough for bringing a collision between both.
    Again, we can safely stop the older run.
\item $R$ and $R'$ are of the same age and $R$ did not yet have an associated merge:
    Then the partner $R^{\star}$ of $R$ has to be located to the right of $R'$.
    But then the subchain between $R$ and $R'$ could not have been a quasi edge (cf.\ Lemma~\ref{lem:older-runs}).
    In this case, we can stop the follower.\qed
\end{enumerate}
\end{proof*}

\begin{proof}[Lemma~\ref{lem:EM1-run}]
Proposition~\ref{prop:quasiedgedestroy} ensures that as long as the good pair is active, the subchain between both remains a quasi edge.
Since a quasi edge only consists of $\EM(0)$ and $\EM(1)$ modules, the actions of Table~\ref{table:robot-actions_a}.$A_{\ref{action:hop}}$ and $A_{\ref{action:hop2}}$ cover these two cases:
The next module ahead of a run is either an $\EM(0)$ module or an $\EM(1)$ module with the same orientation as the module of the run.
If the shape of the next $\EM(1)$ module is a reflection of the run over an axis which is orthogonal to the alignment of the quasi edge, this case is covered by the merge actions $A_{\ref{action:merge}}$.
Together with Lemma~\ref{lem:run-properties}.(\ref{enumitem:disjoint-runs}) and (\ref{enumitem:consec-runs-collisions}), this gives that $R$ (respectively $R'$) can only be stopped by collisions or merges.

If no collision happens until $R$ and $R'$ meet, then it takes at most three executions (because of the pausing) of \texttt{ex\-e\-cute-run()} in order to move them closer together.
Then with Table~\ref{table:robot-actions_a}.$A_{\ref{action:merge}}$, the pair must eventually initiate a merge.

Concerning collisions, the algorithm \texttt{cleanup-runs()} together with Lemma~\ref{lem:mindist} ensures that collisions cannot avoid the pair $R,R'$ to achieve a merge.
Now, (\ref{enumitem:consec-runs-merges}) and (\ref{enumitem:behind-run-merges}) of Lemma~\ref{lem:run-properties} ensure that the associated merge is completely included in the $\subchain{R,R'}$.
\end{proof}

\subsection{Detailed Robot Actions}\label{appssec:tables}
(See next page.)
\begin{table}[H]
\caption{Robot actions $(b)$.}
\label{table:robot-actions_b}
\smallskip
\def\tabularxcolumn#1{m{#1}}
\centering
\begin{tabularx}{\textwidth}{p{1cm} X}
	& \hspace{2.3cm}\textbf{Pattern} \hspace{3.8cm} \textbf{Action} \\
	\hline\hline\\
	\action{action:runinit_a} & \hspace{2.3cm}\includegraphics[scale=0.85]{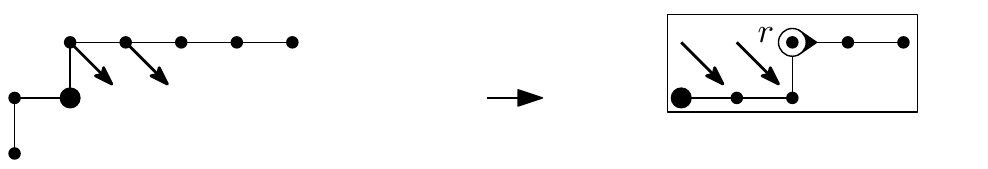} \\
	\action{action:runinit_b} & \hspace{2.3cm}\includegraphics[scale=0.85]{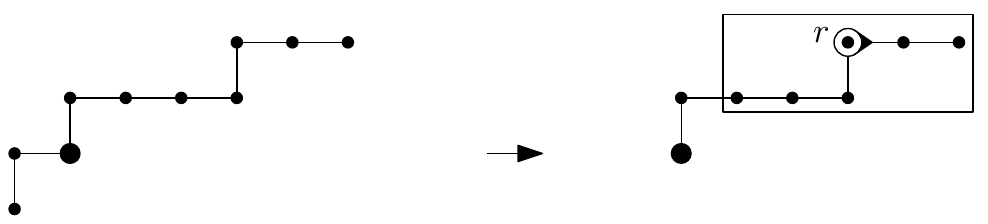} \\
	\action{action:runinit_c} & \hspace{2.3cm}\includegraphics[scale=0.85]{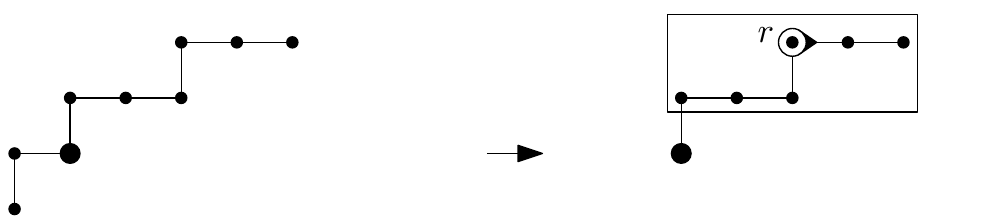} \\
	\\
	& The bold marked robot does not hop. $r$ becomes the new runner. \\
	\hline\\
	\action{action:runinit_spec_a} & \hspace{2.3cm}\includegraphics[scale=0.85]{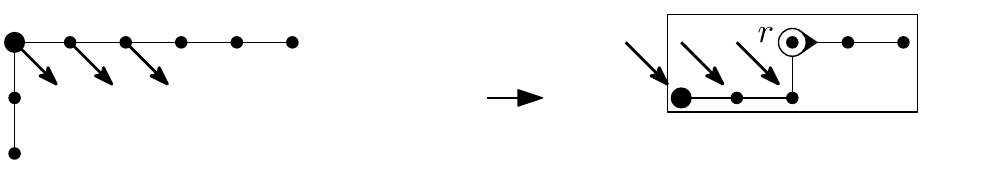} \\
	\action{action:runinit_spec_b} & \hspace{2.3cm}\includegraphics[scale=0.85]{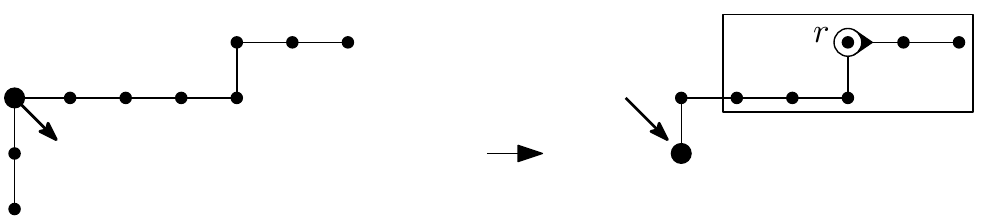} \\
	\action{action:runinit_spec_c} & \hspace{2.3cm}\includegraphics[scale=0.85]{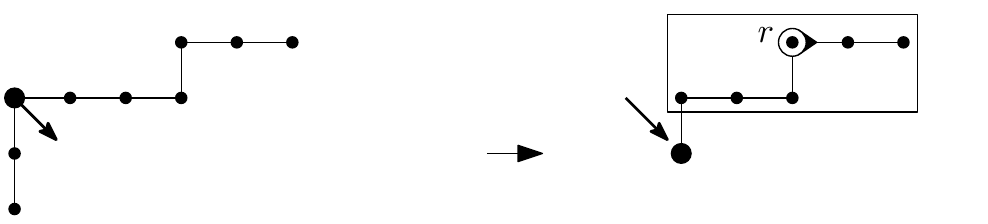} \\
	\action{action:runinit_spec_d} & \hspace{2.3cm}\includegraphics[scale=0.85]{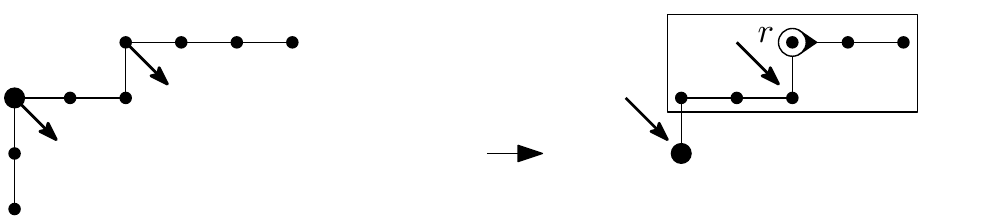} \\
	\action{action:runinit_spec_e} & \hspace{2.3cm}\includegraphics[scale=0.85]{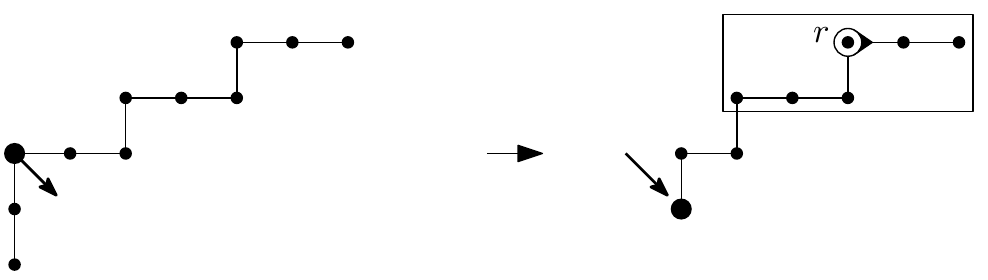} \\
	 \\
	 & The bold marked robot hops. $r$ becomes the new runner. \\
	 \hline
\end{tabularx}
\end{table}

\end{document}